\documentclass[english,journal]{extarticle}
\usepackage{fontenc}
\usepackage[utf8]{inputenc}
\usepackage{color}
\usepackage{subfigure}
\usepackage{graphicx}
\usepackage{babel}
\usepackage{float}
\usepackage{amsthm}
\usepackage{amsmath}
\usepackage{amssymb}
\usepackage{tikz}
\usepackage{fixltx2e}
\usepackage{multirow}
\usepackage{anysize}
\usepackage{parskip}
\usepackage{xargs}[2008/03/08]
\usepackage[unicode=true,pdfusetitle,
 bookmarks=true,bookmarksnumbered=false,bookmarksopen=false,
 breaklinks=false,pdfborder={0 0 0},backref=false,colorlinks=true]
 {hyperref}
\marginsize{1in}{1in}{1in}{1in}
\makeatletter

\floatstyle{ruled}
\newfloat{algorithm}{tbp}{loa}
\providecommand{\algorithmname}{Algorithm}
\floatname{algorithm}{\protect\algorithmname}

  \theoremstyle{plain}
  \newtheorem{thm}{\protect\theoremname}
  \providecommand{\cnjname}{Conjecture}
  \theoremstyle{plain}
  \newtheorem{cnj}{\protect\cnjname}
  \theoremstyle{plain}
  \newtheorem{lem}{\protect\lemmaname}
  \theoremstyle{definition}
  
  \theoremstyle{plain}
  \newtheorem{cor}{\protect\corollaryname}
 \ifx\proof\undefined\
   \newenvironment{proof}[1][\proofname]{\par
     \normalfont\topsep6\p@\@plus6\p@\relax
     \trivlist
     \itemindent\parindent
     \item[\hskip\labelsep
           \scshape
       #1]\ignorespaces
   }{%
     \endtrivlist\@endpefalse
   }
   \providecommand{\proofname}{Proof}
 \fi
  \theoremstyle{plain}

\usepackage{algpseudocode}
\usepackage{hhline}
\usepackage{array,booktabs}

\AtBeginDocument{
  
}

\makeatother

\providecommand{\corollaryname}{\inputencoding{latin9}Corollary}
\providecommand{\definitionname}{\inputencoding{latin9}Definition}
\providecommand{\lemmaname}{\inputencoding{latin9}Lemma}
\providecommand{\propositionname}{\inputencoding{latin9}Proposition}
\providecommand{\theoremname}{\inputencoding{latin9}Theorem}

\newcommand{\mathref}[2]{\hyperref[msym:#1]{#2}}

\usepackage{xy}
\xyoption{matrix}
\xyoption{frame}
\xyoption{arrow}
\xyoption{arc}

\usepackage{ifpdf}
\ifpdf
\else
\PackageWarningNoLine{Qcircuit}{Qcircuit is loading in Postscript mode.  The Xy-pic options ps and dvips will be loaded.  If you wish to use other Postscript drivers for Xy-pic, you must modify the code in Qcircuit.tex}
\xyoption{ps}
\xyoption{dvips}
\fi

\entrymodifiers={!C\entrybox}

\newcommand{\ket}[1]{{\left\vert{#1}\right\rangle}}
\newcommand{\qw}[1][-1]{\ar @{-} [0,#1]}
\newcommand{\qwx}[1][-1]{\ar @{-} [#1,0]}
\newcommand{\control}{*!<0em,.025em>-=-<.2em>{\bullet}}
\newcommand{\ctrl}[1]{\control \qwx[#1] \qw}
\newcommand{\targ}{*+<.02em,.02em>{\xy ="i","i"-<.39em,0em>;"i"+<.39em,0em> **\dir{-}, "i"-<0em,.39em>;"i"+<0em,.39em> **\dir{-},"i"*\xycircle<.4em>{} \endxy} \qw}
\newcommand{\lstick}[1]{*!R!<.5em,0em>=<0em>{#1}}
\newcommand{\Qcircuit}{\xymatrix @*=<0em>}
\begin{document}

\title{Optimal and asymptotically optimal NCT reversible circuits by the gate types}

\author{\small{ Dmitri Maslov$^{1,2}$} \\
{\small\it $^1$ National Science Foundation, Arlington, VA, USA} \\
{\small\it $^2$ QuICS, University of Maryland, College Park, MD, USA} \\
{\small\tt \href{mailto:dmitri.maslov@gmail.com}{dmitri.maslov@gmail.com}}\\
}

\maketitle

\begin{abstract}
We report optimal and asymptotically optimal reversible circuits composed of NOT, CNOT, and Toffoli (NCT) gates, keeping the count by the subsets of the gate types used.  This study fine tunes the circuit complexity figures for the realization of reversible functions via reversible NCT circuits.  An important consequence is a result on the limitation of the use of the $T$-count quantum circuit metric popular in applications. 
\end{abstract}

\section{Introduction} 
Reversible circuits are important parts of quantum algorithms.  Grover's oracles, integer and finite field arithmetic operations (used in Shor-type discrete logarithm quantum algorithms), as well as numerous types of Boolean operations over quantum registers are all examples of the reversible circuits.  Consequently, the study of reversible circuits and their complexities is important in understanding the complexity of quantum circuits and algorithms, as well as for the efficient implementation of quantum algorithms.  
 
In this paper, we study reversible circuits over the gate library consisting of the NOT, the CNOT, and the Toffoli gates, also commonly referred to as NCT circuits.  The individual gates are defined via the logical transformations they perform over Boolean variables, as follows:
\begin{itemize}
\item NOT gate, $\textsc{NOT}(a): a \mapsto a \oplus 1$;
\item CNOT gate, $\textsc{CNOT}(a;b): (a,b) \mapsto (a, b \oplus a)$;
\item Toffoli gate, $\textsc{TOF}(a,b;c): (a,b,c) \mapsto (a,b,c \oplus ab)$.
\end{itemize}
A reversible circuit is the string of gates, read left to right.  In addition to the primary variable inputs, a reversible circuit may have constant inputs, carrying a constant value of either a zero or a one.  Those additional inputs are called ancillae.  They can be a useful resource, as they provide additional space for the computations. 

A reversible function of $n$ Boolean variables, $f(x)=f(x_1,x_2,...,x_n)=(f_1(x_1,x_2,...,x_n),$ \linebreak $f_2(x_1,x_2,...,x_n),...,f_n(x_1,x_2,...,x_n))$ is the bijective mapping of the $n$-dimensional Boolean cube into itself.  There are three possible notions of what it means to implement a reversible function by a reversible circuit, that we list next, appearing in the weakest to the strongest form.
\begin{itemize}
\item[W.] Weak. The reversible circuit computes a set of Boolean functions, and among them, are all $n$ outputs of the desired reversible function $f$.
\item[I.] Intermediate.  The reversible circuit computes the mapping $(x,y) \mapsto (x,y\oplus f(x))$, where $x$, $y$, and $f(x)$ are $n$-bit registers, and the EXOR operates bitwise.  In addition, some ancillae may be used, but their values are returned to the original state.
\item[S.] Strong. The circuit implementing $f$ performs the mapping $x \mapsto f(x)$. Some ancillae may be used, but their values are returned to the original state. 
\end{itemize} 

Stronger notions of the implementability can be used straightforwardly to compute the weaker notions; some CNOTs may be required.  Weaker notions can too be used to compute the stronger notions.  To construct the intermediate implementation using the weak implementation, apply the weak circuit $A$, EXOR useful outputs, $f(x)$, to the new register $y$ via the use of $n$ CNOTs, and then apply $A^{-1}$.  Recall that the circuit $A^{-1}$ may be obtained from the NCT circuit $A$ via inverting the order of gates in $A$.  As a result, the intermediate implementation can be constructed with at most twice the number of gates in the weak implementation, plus $n$ CNOT gates.  Incidentally, same procedure can be applied to the strong implementation to obtain the intermediate implementation from it. 

To obtain the strong implementation from the intermediate implementation, take two circuits---circuit $B$ computing $(x,y) \mapsto (x,y\oplus f(x))$ and circuit $C$ computing $(x,y) \mapsto (x,y\oplus f^{-1}(x))$.  To obtain the transformation $x \mapsto f(x)$ start with the $2n$-bit register $(x,0)$, apply $B$ to it to transform it into $(x,f(x))$, then SWAP first and second $n$-bit registers to obtain $(f(x),x)$, and finally apply $C$ to obtain $(f(x), x \oplus f^{-1}(f(x))) = (f(x), x \oplus x) = (f(x),0)$.  Discarding the mention of ancillae, the aggregate transformation can now be described as $x \mapsto f(x)$.  A further in-depth study of the relation between the intermediate and strong forms can be found in \cite{ar:kkvb}.

Neither of the above constructions affects asymptotics in the case when we are concerned with upper bounds on the resources required to obtain the most difficult function.  Indeed, if the upper bound on the resource count used by $B$, performing the mapping $(x,y) \mapsto (x,y\oplus f(x))$, is $I(n)$, same number, $I(n)$, applies to upper bound the cost of $C$, performing the mapping $(x,y) \mapsto (x,y\oplus f^{-1}(x))$.  Next we summarize how the following notions are related: $W(n)$ the cost of the weak implementation of an arbitrary $n$-bit reversible function, $I(n)$ the cost of the intermediate implementation of an arbitrary $n$-bit reversible function, and $S(n)$ the cost of the strong implementation of an arbitrary $n$-bit reversible function:
\begin{itemize}
\item[WI:] $W(n) \leq I(n)$;
\item[IS:] $I(n) \leq 2\cdot S(n) + n\cdot Cost(\textsc{CNOT})$;
\item[WS:] $W(n) \leq S(n)$;
\item[IW:] $I(n) \leq 2\cdot W(n) + n\cdot Cost(\textsc{CNOT})$;
\item[SI:] $S(n) \leq 2\cdot I(n) + 3n\cdot Cost(\textsc{CNOT})$;
\item[SW:] $S(n) \leq 4\cdot W(n) + 5n\cdot Cost(\textsc{CNOT})$. 
\end{itemize}
In the above, we relied on the common notion that the SWAP gate can be implemented via three CNOTs.  We conclude that in the case when $W(n)$, $I(n)$, and $S(n)$ are at least linear in $n$, asymptotic optimality of any one of them implies the asymptotic optimality of all other types of implementations. 

The reason to have multiple definitions of computability is rooted in the observation that the weak notion would be the one expected in the scenario when reversible circuits are interesting in the context of their own.  However, reversible circuits are most often viewed in the broader context of quantum computing \cite{bk:nc}.  From the point of view of quantum computations, the weak notion of computability by a reversible circuit may give rise to the unwanted entanglement residing on those (qu)bits carrying partial results of the computation.  The intermediate notion removes the concern of the unwanted entanglement residing on the partially computed outputs, and furthermore is used broadly within the context of quantum algorithms.  Should the strong notion of computability be required, it is possible to obtain it too, without affecting the asymptotic optimality.  Therefore, from the point of view of this paper, we will be satisfied with any one type of implementation. 

Define $L_{a,b,c}(n,g)$, where $a,b,c \in \{0,1\}$, and $g$ is a positive integer, to be the smallest cost of the circuit implementation of the most expensive reversible function of $n$ variables realized by the NCT circuit using at most $g$ input constants.  The input constants are allowed to take values $0$ or $1$.  The circuit cost is calculated as the sum across all gates participating in the given circuit, where the NOT gate is counted with the weight $a$, the CNOT gate is counted with the weight $b$, and the Toffoli gate is counted with the weight $c$.  We note that $L_{a,b,c}(n,0)$ does not exist, as no odd permutation may be synthesized via an NCT circuit without using an additional ancilla \cite{ar:spmh}; this explains why we chose $g>0$ in the definition of $L_{a,b,c}(n,g)$.  $L_{a,b,c}(n)$ furthermore reports the best cost NCT implementation of the function that is most difficult to obtain via its circuit realization in the scenario where the use of an arbitrary number of ancillae is allowed. 

Of the 8 possible choices for parameters $a,b,c$ in $L_{a,b,c}(n,g)$, some carry a special meaning.  For instance, $L_{1,0,0}(n,g)$ determines the maximal number of the NOT gates required by reversible circuits.  We will later show that this number is zero, meaning NOT gates by themselves are not very helpful, as their use can be avoided.  $L_{0,0,1}(n,g)=L_{0,0,1}(n)$, when $g$ is allowed to be arbitrarily high, such as to not limit the space used by the computation, can be viewed as the multiplicative, or otherwise, non-linear cost of the reversible functions.  It is furthermore closely related to the $T$-count circuit metric in quantum circuits, as both ignore the effects of the cost of linear (with respect to EXOR) reversible transformations.  The study of lower and upper bounds on $L_{0,0,1}(n)$, as well as the discussion of the implications, is the main focus of this paper. 

\subsection{Motivation}
In this paper we study the problem of minimizing the gate counts by their type in reversible circuits with NOT, CNOT, and Toffoli gates.  This study is motivated by the relative hardness of constructing the Toffoli gate compared to the effort required to obtain the CNOT, and the relative hardness of CNOT compared to NOT.  

First, compare the implementation costs of NOT and CNOT.  Both are Clifford gates, therefore, on the logical level they are likely to be transversal.  This means that both gates are implemented via a set of NOT, and, respectively, CNOT gates, applied to the physical-level qubits.  On the physical level, a NOT gate is often less expensive than the CNOT gate.  Frequently, this is due to the two-qubit gates taking more effort to implement than any of the single-qubit gates.  It is not uncommon for a CNOT gate to be $20$ times more resource demanding compared to the NOT gate. 

% To give a specific example, within trapped ions quantum information processing approach, the NOT gate takes time 10$\mu s$, whereas the CNOT gate has a duration of 275$\mu s$.  Apart from taking more time, the CNOT gate also introduces a significantly higher error compared to the NOT gate due to the imperfections of the controlling apparatus and a higher number of pulses required to obtain it.

Next, compare the cost of the Toffoli gate to the cost of the CNOT gate.  Consider quantum logical-level circuits.  Often, all Clifford gates, CNOT included, are relatively easy to implement on the logical level.  In contrast, the Toffoli gate, being a transformation outside the Clifford group, is more difficult to obtain.  Assuming the non-Clifford gate provided by the fault tolerance approach selected is the so-called $T$ gate, the Toffoli gate may be implemented as a circuit with 2 Hadamard gates, 7 $T/T^\dagger$ gates, and 6 CNOT gates \cite{ar:ammr}.  Discarding the cost of the Hadamard gate, and taking the sum of the costs of the remaining gates in this circuit implementation, we obtain $Cost(\textsc{TOF}) = 6\cdot Cost(\textsc{CNOT}) + 7\cdot Cost(T/T^\dagger)$.  While it was shown that 6 CNOT gates are required to implement the Toffoli gate as a circuit over the library including arbitrary single-qubit and CNOT gates \cite{ar:sm}, the number 7 of $T/T^{\dagger}$ gates has been obtained via a computer search \cite{ar:ammr} and in principle could be reducible (and, in fact, it is when additional resources are available \cite{ar:j}).  The known way of implementing the fault-tolerant gate $T$ requires state distillation and then its teleportation.  The teleportation is achieved via the use of the single logical CNOT gate, relying on the well-known teleportation circuit, and therefore it is not resource demanding.  The state distillation relies on a nested application of the 15-qubit Hamming code \cite{ar:bk}.  Physical parameters of the quantum information system used and the overall length of the desired computation play a determining role in deciding on the details of the protocol and the complexity of implementing the $T$ gate.  Assuming the distillation depth of 2, which appears to be a practical choice for scalable computations, the cost of implementing the $T$ gate is roughly 50 times that of the logical CNOT.  At which point, the cost of the Toffoli gate expressed in the units corresponding to the cost of the CNOT gates becomes roughly $6\cdot 1+7\cdot 50=356$, being, in practical terms, a large number.  While the number $356$ may itself be possible to reduce ({\em e.g.}, outsource the ancilla production to before the desired computation), it is likely that the Toffoli gate will remain substantially (provably, at least 6 times over arbitrary single-qubit gates and the CNOT \cite{ar:sm}) more expensive than the (nearest-neighbour) CNOT gate.  

One other resource that can be useful for efficient implementation of reversible functions is ancillary (qu)bits.  It may be difficult to compare the cost of arbitrary gates to that of ancillary qubits directly, as these are, strictly speaking, resources of a different kind.  However, we will next evaluate the relation between the cost of a logical ancilla qubit and that of logical NOT/Toffoli gates to conclude that, within known fault tolerant approaches, the ancilla qubit is substantially more expensive than the NOT gate, and substantially less expensive than the Toffoli gate.  First, to obtain a logical ancilla qubit in the state $\ket{0}$ or $\ket{1}$ we need a physical space (qubits).  This physical space needs to be prepared in the respective encoded logical state.  The procedure accomplishing it can be described as a (physical-level) Clifford circuit.  As a result, one may expect a number of the physical-level CNOT gates to be applied.  Thus, recalling earlier discussions, it can be expected that the cost of this operation exceeds that of the logical NOT gate.  Second, the Toffoli gate relies on the seven $T/T^\dagger$ gates that themselves are obtained with the use of state distillation, which includes a nested application of the circuits implementing a Clifford unitary.  Logical $\ket{0}$ or $\ket{1}$ state preparation, on the other hand, uses only one Clifford circuit designed to prepare a state, as opposed to a whole unitary, which is expected to be much simpler to accomplish. 

The cost of an ancilla can thus be said to be roughly similar to that of the CNOT gate.  As such, we will disregard the cost of ancilla every time we exclude the cost of the CNOT gates from the circuit cost figure ($L_{a,b,c}(n,g)$ with $b=0$).  In the scenario when we include the CNOT count in the overall calculation, we may consider accounting for ancilla as well.  There are only two interesting cases to consider, $L_{1,1,1}(n,g)$ and $L_{0,1,1}(n,g)$.  This is because, as shown later, the remaining two relevant complexity figures, $L_{0,1,0}(n,g)=0$ when $g \geq 1$, and $L_{1,1,0}(n,1)=1$ and $L_{1,1,0}(n,g)=0$ when $g \geq 2$ carry small values, and there is little interplay between the value of the complexity function $L$ and the number of ancillae.  The case $L_{0,1,1}(n,g)$ may furthermore be reduced to $L_{1,1,1}(n,g)$.  Indeed, constructing the lower bound involves the counting argument, and the asymptotics of the number of transformations achievable by cost-1 circuits over the library with ``free'' NOT gates is the same as the asymptotics of the number of transformations achievable by cost-1 circuits when the NOT gates are counted towards the cost figure.  This means that essentially the same lower bounds will apply to both $L_{1,1,1}(n,g)$ and $L_{0,1,1}(n,g)$.  An upper bound for the quantity $L_{1,1,1}(n,g)$ can furthermore be used directly to upper bound the quantity $L_{0,1,1}(n,g)$.  As a result, of the two quantities, $L_{1,1,1}(n,g)$ and $L_{0,1,1}(n,g)$, only one, $L_{1,1,1}(n,g)$, may be studied, with the results transferable to $L_{0,1,1}(n,g)$.  

If we allow arbitrary ancillae, a classical Boolean circuit complexity result stating that any Boolean function can be implemented using at most $O(\frac{2^n}{n})$ NOT/OR/AND classical gates \cite{ar:l} can be used to upper bound the number of NOT/OR/AND gates in a classical irreversible circuit by $O(2^n)$ for every reversible function.  Making this latter classical circuit reversible may increase the number of gates used by some constant factor, and will not require more than $O(2^n)$ ancillae.  This leads to the upper bound of the form $L_{1,1,1}(n,C_22^n) \lesssim C_32^n$ for a proper choice of constants $C_2$ and $C_3$.  We can furthermore lower bound $L_{1,1,1}(n,C_22^n)$ by the quantity $C_12^n$, for a proper choice of the constant $C_1$, via applying the simple counting argument, \cite[Lemma 8]{ar:spmh}, to obtain asymptotic optimality, $C_12^n \lesssim L_{1,1,1}(n,C_22^n) \lesssim C_32^n$.  It is interesting to study how the value of $L_{1,1,1}(n,g)$ function changes when the number of ancillae $g$ is increased from 1 to $C_22^n$, however, such a study is outside the scope of this paper.  For the rest of the paper, we will restrict the number of ancillae to a constant when considering the values of the $L$ function with the CNOT gate count included.

\subsection{Previous work}
The topic of the complexity of NCT realizations of reversible functions has been studied extensively.  Within the terminology introduced above, previous literature encounters the following results. \cite{ar:spmh} reports upper and lower bounds of the following form: $\frac{n2^n}{3\log_2{n}} \lesssim L_{1,1,1}(n,1) \lesssim 9n2^n$, $L_{0,1,1}(n,1) \lesssim 9n2^n$, and $L_{0,0,1}(n,1) \lesssim 9n2^n$.  The upper bounds were improved to $L_{1,1,1}(n,1) \lesssim 5n2^n$, $L_{0,1,1}(n,1) \lesssim 5n2^n$, and $L_{0,0,1}(n,1) \lesssim 3n2^n$ in \cite{ar:mmd}, and then to $L_{1,1,1}(n,1) \lesssim 4.5n2^n$, $L_{0,1,1}(n,1) \lesssim 4n2^n$, and $L_{0,0,1}(n,1) \lesssim 2n2^n$ in \cite{pc:s, ar:szss}.  Finally, \cite{arXiv:1504.06876} reports improved upper bounds of the following form: $L_{1,1,1}(n,1) \lesssim \frac{48n2^n}{\log_2{n}}$, $L_{0,1,1}(n,1) \lesssim \frac{40n2^n}{\log_2{n}}$, and $L_{0,0,1}(n,1) \lesssim \frac{32n2^n}{\log_2{n}}$.  Technically, to apply the result \cite{arXiv:1504.06876}, that was itself developed to handle even permutations, we need to mention that any odd permutation can be reduced to an even permutation via multiplying it by the maximal size multiple control Toffoli gate (one spanning all $n$ bits).  This latter multiple control Toffoli gate requires a linear number of 3-bit Toffoli gates to be implemented as an NCT circuit \cite{ar:bbcd}, and therefore does not affect the advertised asymptotics.  The multiplicative constant separating the best known lower \cite{ar:spmh} and upper \cite{arXiv:1504.06876} bounds for $L_{1,1,1}(n,1)$ is about a hundred, however, it is no more a function of $n$, allowing to conclude that the asymptotic optimality has been established. 

The topic of gate-specific Boolean circuit complexities has also been studied in the literature.  Specifically, \cite{ar:n} studied the complexity of formulas, circuits, and contact relays, implementing a Boolean single-output function, including the multiplication cost (in the basis with Boolean multiplication and addition modulo-2, where the addition can be used for free) of Boolean single-output functions.  The upper bound developed in \cite{ar:n} on the number of multiplications required, $2^{n/2}+o(2^{n/2})$, may be applied to the reversible case to obtain the upper bound of $n2^{n/2}$, in the leading order, of the Toffoli gates per a reversible function.  However, in this paper, we are able to show a better upper bound of $\frac{3}{\sqrt{2}}\sqrt{n}2^{n/2}$ in the leading order via a direct construction.  The difference between classical and reversible logic cases is furthermore not only in the different number of outputs ($n$ VS 1, assuming $n$ input bits) that need to be computed by the reversible circuit simultaneously, but also in the ancillae management, that is of no importance in classical circuits.  These differences are substantial in that they appear to prohibit a direct transfer of the results from standard Boolean logic to the reversible function/circuit scenario.

\section{$L_{a,b,c}(n,g)$: practice-motivated and other cases}
We have previously established that $Cost(\textsc{NOT}) < Cost(\textsc{CNOT}) < Cost(\textsc{TOF})$.  Therefore, there is most value in studying circuit complexities by the gate types in the scenario discarding the costs of the simpler resources first.  These are the cases of $L_{0,1,1}(n,g)$ and $L_{0,0,1}(n)$.  Recall that asymptotically optimal lower and upper bounds in the scenario $L_{1,1,1}(n,1)$ have already been obtained by the previous authors.  We study $L_{0,1,1}(n,g)$ and $L_{0,0,1}(n)$ in the following sections.  In this section, we study optimal and asymptotically optimal circuits in the remaining four scenarios ($L_{0,0,0}(n,g)$ is trivial), $L_{1,0,0}(n,g)$, $L_{1,1,0}(n,g)$, $L_{0,1,0}(n,g)$, and $L_{1,0,1}(n,g)$. 

Consider $L_{1,0,0}(n,g)$, counting the number of NOT gates in reversible NCT circuits.  It may be easily established that $\forall g$ $L_{1,0,0}(n,g)=0$.  Indeed, assign the value of $1$ to the ancillary qubit $y$.  Every time a $\textsc{NOT}(x_i)$ gate is used by the synthesis algorithm reported in, {\em e.g.}, \cite{co:mmd} (a \textsc{NOT} gate is applied at most once to each primary input in the beginning of the circuit), replace it with the $\textsc{CNOT}(y;x_i)$.  The result of this modification of the synthesis algorithm in \cite{co:mmd} is no NOT gates are used, and therefore $L_{1,0,0}(n,g)=0$.  Similarly, if one were to discard the cost of the Toffoli gates, and study the NOT/CNOT cost, the following statement is true: $\forall g \geq 2$ $L_{1,1,0}(n,g)=0$.  The case of $L_{1,1,0}(n,1)$ is somewhat more complex and requires an explicit proof. 

\vspace{2mm}\begin{lem}
$L_{1,1,0}(n,1)=1.$
\end{lem}
\begin{proof}
{\bf Lower bound.} There are two cases to consider: on the input side we have either (a) values $x_1,x_2,...,x_n,0$ or (b) values $x_1,x_2,...,x_n,1$, where $x_1,x_2,...,x_n$ are primary inputs.  To keep the NOT and CNOT count to zero, we may use only the Toffoli gates.  We next prove that there exists a reversible function that may not be computed by such a circuit.  The impossibility implies that at least one NOT or CNOT gate needs to be used.  In the proof of the upper bound that follows we will furthermore show that one NOT gate does suffice, leading to the desired equality.  To prove that we need a NOT or a CNOT gate, apply a series of Toffoli gates, and observe that the number of components with term $1$ in their PPRM expansion (Positive Polarity Reed-Muller expansion, see (\ref{eq:pprm})) remains constant (zero in the case (a) and one in the case (b)). Indeed, to include the term $1$ in some component $y$ not yet containing a $1$ in its PPRM expansion, this bit needs to be a target of a Toffoli gate.  Suppose this is the Toffoli gate with controls $y_1$ and $y_2$.  To obtain term $1$ in the PPRM expansion of the product $y_1y_2$ both $y_1$ and $y_2$ need to contain the term $1$ in their PPRM expansion.  However, we have at most zero (case (a)) or one (case (b)) components containing the term $1$, but not two that we need.  To conclude the proof observe that the function $(x_1 \oplus 1, x_2 \oplus 1, x_3, ..., x_n)$ is reversible and has two of its components contain term $1$ in their PPRM expansion.  This function cannot be generated by the Toffoli gates alone.

{\bf Upper bound.} We modify the algorithm reported in \cite{co:mmd} to implement any reversible function using only the Toffoli gates and at most one NOT gate, when only a single ancilla is available.  First, choose the value of the ancilla to be $1$. Take a reversible function and synthesize a reversible circuit implementing it using the Toffoli gates and no more than one NOT gate.

Step 1 of the basic algorithm \cite{co:mmd} prescribes the use of NOT gates for every bit where $f(0) \neq 0$.  Without loss of generality, assume these are the first $k$ bits.  Then, replace the circuit composed with these $k$ NOT gates, $\textsc{NOT}(x_1)\textsc{NOT}(x_2)...\textsc{NOT}(x_k)$, with the functionally equivalent circuit $\textsc{TOF}(1,x_1;x_2)$ $\textsc{TOF}(1,x_1;x_3)...\textsc{TOF}(1,x_1;x_k) \textsc{NOT}(x_1)\textsc{TOF}(1,x_1;x_k)...\textsc{TOF}(1,x_1;x_3)\textsc{TOF}(1,x_1;x_2)$, such as illustrated next ($k=3, n=4$):
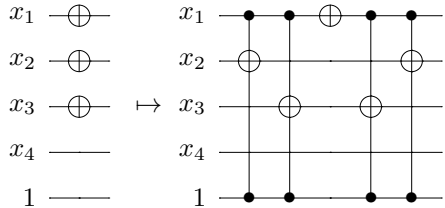
\begin{eqnarray*}\label{circ:tof1}
\Qcircuit @C=0.7em @R=0.9em @!R {
\lstick{x_1}	 & \targ	& \qw \\
\lstick{x_2}	 & \targ	& \qw \\
\lstick{x_3}	 & \targ	& \qw \\
\lstick{x_4}	 & \qw		& \qw \\
\lstick{1} 		 & \qw		& \qw
}
&
\raisebox{-3.4em}{$\mapsto$\hspace{4mm}}
&
\Qcircuit @C=0.7em @R=0.9em @!R {
\lstick{x_1}	 & \ctrl{1}	& \ctrl{2} 	& \targ & \ctrl{2} 	& \ctrl{1}	& \qw \\
\lstick{x_2}	 & \targ	& \qw 		& \qw 	& \qw 		& \targ		& \qw \\
\lstick{x_3}	 & \qw		& \targ 	& \qw 	& \targ 	& \qw		& \qw \\
\lstick{x_4}	 & \qw		& \qw		& \qw 	& \qw		& \qw		& \qw \\
\lstick{1} 		 & \ctrl{-3}& \ctrl{-2}	& \qw 	& \ctrl{-2}	& \ctrl{-3}	& \qw
}
\end{eqnarray*}
Observe that the circuit on the right hand side uses only one NOT gate, at which point, the circuit we have synthesized thus far has the NOT/CNOT cost of one.

Steps $2$ to $2^n$ use only CNOT, Toffoli, and multiple control Toffoli gates.  We replace each $\textsc{CNOT}(x_i;x_j)$ with $\textsc{TOF}(1,x_i;x_j)$ and replace each multiple control Toffoli gate with its Toffoli gate realization, \cite[Lemmas 7.2, 7.3]{ar:bbcd}.  In particular, when needed, we break down the given multiple control Toffoli gate into four smaller multiple control Toffoli gates using Lemma 7.3, and then apply Lemma 7.2 to decompose all smaller multiple control Toffoli gates into 3-bit Toffoli gates.  The resulting circuit may use the anicilla bit that we have available, carrying the known value $1$.  As such, steps $2$ to $2^n$ use no NOT or CNOT gates, and the overall NOT/CNOT gate count is 1.  
\end{proof}

Next, consider $L_{0,1,0}(n,g)$, counting the number of CNOTs in the reversible NCT circuits.  $\forall g$ $L_{0,1,0}(n,g)=0$, since one may select the ancillary qubit $y$ to carry the value 1, and replace each $\textsc{CNOT}(x_i;x_j)$ used \cite{co:mmd} with $\textsc{TOF}(y,x_i;x_j)$. 

Finally, the case $L_{1,0,1}(n,g)$ may be reduced to $L_{0,0,1}(n,g)$, considered in Section \ref{sec:001}.  This is because in the presence of ``free'' CNOT gates, each $\textsc{NOT}(x_i)$ gate can be replaced by a $\textsc{CNOT}(1;x_i)$, being an implementation of NOT with zero cost.

\section{$L_{0,1,1}(n,g)$}

\vspace{2mm}\begin{lem}
$\frac{n2^n}{3 \log_2{n}} \lesssim L_{0,1,1}(n,1) \lesssim \frac{40n2^n}{\log_2{n}}$.
\end{lem}

\begin{proof}
{\bf Lower bound.}  We rely on \cite[Lemma 8]{ar:spmh} to construct the lower bound.  Specifically, \cite[Lemma 8]{ar:spmh} states that the quantity $\frac{\log_2{G}}{\log_2{b}}$, where $G$ is the size of the set of functions being computed, and $b$ is the number of different cost-one circuits, lower bounds the cost of the circuits that implement an arbitrary function.  In our calculations, $G = \frac{2^n!}{2^{n}}$, since each function may be implemented by a circuit up to the possible ``free'' negation of all output side wires.  The number of the different cost-one circuits (Toffoli and CNOT gates with arbitrary NOT gates on the input side) is $4n^3 + o(n^3)$.  The ratio $\frac{\log_2{G}}{\log_2{b}}$ is then equal to $\frac{n2^n}{3 \log_2{n}}$ up to the lower degree additive terms. 

{\bf Upper bound.} $L_{0,1,1}(n,1) \lesssim \frac{40n2^n}{\log_2{n}}$ was shown in \cite{arXiv:1504.06876}.
\end{proof}

\section{$L_{0,0,1}(n)$}\label{sec:001}

\subsection{Lower bound}
In the next we will show that the number of the Toffoli gates required to implement an arbitrary reversible $n$-bit function $f(x_1,x_2,...,x_n)=(f_1(x_1,x_2,...,x_n),f_2(x_1,x_2,...,x_n),...,f_n(x_1,x_2,...,x_n))$ of $n$ primary inputs $x_1,x_2,...,x_n$ with $n$ primary outputs $f_1,f_2,...,f_n$, $L_{0,0,1}(n)$, is at least $\sqrt{n} 2^{n/2}+o(\sqrt{n} 2^{n/2})$.  To accomplish this consider a circuit with $h$ Toffoli gates.  We will number and refer to the Toffoli gates within the circuit as $\textsc{TOF}_1, \textsc{TOF}_2, ...., \textsc{TOF}_h$, in the order they appear in the circuit (first to last).  We furthermore break bits/wires in the circuit into smaller chunks.  In particular, denote $w_a$ to be an uninterrupted piece of wire between some two gates.  The values carried by those pieces of wire are equal to the respective primary input/constant between that input/constant and the first gate applied; the values of the pieces of wire in the middle of the circuit and the output depend on which gates were applied by the circuit.  For instance, for a $\textsc{TOF}_i(w_a,w_b;w_c)$ over three input-side pieces of wire, being two input controls $w_a$ and $w_b$, and one input target $w_c$, and three pieces of wire on the output side, $w_d, w_e,$ and $w_f$, the values carried by the pieces of wire are related by the following formulas $w_d=w_a$, $w_e=w_b$, and $w_f=w_c \oplus w_aw_b$.  We will furthermore denote the Boolean product computed by the Toffoli gate $\textsc{TOF}_i$ and EXORed into its target, $w_aw_b$, as $Prod(\textsc{TOF}_i)$. 
 
\vspace{2mm}\begin{lem}\label{lem1}
In a reversible NCT circuit with $h$ Toffoli gates each value carried by a piece of wire $w_a$ can be written as a linear sum $LS(w_a) = \bigoplus_{i=1}^{h}c_i Prod(\textsc{TOF}_i) \oplus l(x)$, where $c_i \in \{0,1\}$ and $l(x)$ is a linear function of primary inputs. 
\end{lem}
\begin{proof}
To construct the linear representation advertised in the statement of Lemma, start with an empty linear sum, $LS=0$ and traverse the given piece of wire $w_a$ back towards the beginning of the circuit until the terminal pieces of wire are found.  Terminal pieces of wire include all primary inputs and all input constants.  Define $S:=\{w_a\}$.  The set $S$ contains all pieces of wire we need to look at.  We next traverse the circuit and compute $LS$.
\begin{itemize}
\item For a piece of wire $w_a \in S$ and upon finding a NOT gate with input $w_b$ and output $w_a$, replace $w_a$ with $w_b$ in the set $S$, and replace $LS$ with $LS \oplus 1$.
\item For a piece of wire $w_a \in S$ and upon finding a CNOT gate with input control $w_b$, input target $w_c$, and output target $w_a$, replace $S$ with $S \setminus \{w_a\} \cup \{w_b, w_c\}$.
\item For a piece of wire $w_a \in S$ and upon finding a CNOT gate with input control $w_b$ and output control $w_a$, replace $S$ with $S \setminus \{w_a\} \cup \{w_b\}$.
\item For a piece of wire $w_a \in S$ and upon finding a Toffoli gate $\textsc{TOF}_i$ with input target $w_b$ and output target $w_a$, replace $w_a$ with $w_b$ in $S$, and replace $LS$ with $LS \oplus Prod(\textsc{TOF}_i)$.
\item For a piece of wire $w_a \in S$ and upon finding a Toffoli gate $\textsc{TOF}_i$ with output control $w_a$, and input control $w_b$ on the same bit as $w_a$, replace $w_a$ with $w_b$ in the set $S$.
\item For a piece of wire $w_a \in S$ carrying a primary input signal $x_k$ or an input constant $Const$, remove $w_a$ from $S$ and replace $LS$ with $LS \oplus x_k$ or $LS \oplus Const$, correspondingly. 
\end{itemize}
Observe that $S$ may already include a $w_b$ when we try to add $w_b$ to it.  It is safe to keep only those pieces of wire in the set $S$ that appear in it with odd multiplicity.  The above algorithm terminates in time at most linear in the number of gates in the circuit.
\end{proof}

\vspace{2mm}\begin{thm}
$\sqrt{n}2^{n/2} \lesssim L_{0,0,1}(n)$.
\end{thm}

\begin{proof}
We will apply the counting argument to obtain the desired lower bound.  The key in successfully using this strategy is to encode reversible functions via such a combinatorial structure that tightly (the encoding must not be too wasteful, such as to affect asymptotics) encodes different functions via their circuit representation and it is easy to either count or upper bound the number of such structures, as parametrized by the number of the Toffoli gates used.  In such case, the number $h$, a parameter in the formula counting the number of such structures, may be used to lower bound the value $L_{0,0,1}(n)$, as $h$ needs to be sufficiently high before the number of combinatorial structures becomes large enough to contain $2^n!$ instances, where $2^n!$ is the number of reversible functions of $n$ inputs.

We next map reversible NCT circuits over $n$ primary inputs containing $h$ Toffoli gates into directed acyclic graphs with edge and vertex labels.  

{\bf Vertices and edges.} The set of vertices is a union of two sets, $T$ and $F$.  The set $T$ contains $h$ elements, $\{T_1,T_2,..., T_h\}$, each corresponding to the respective Toffoli gate in the original circuit, $\{\textsc{TOF}_1,\textsc{TOF}_2,..., \textsc{TOF}_h\}$.  The Toffoli gates within the original circuit are numbered in the ascending order as they appear in the circuit.  The set $F$ consists of $n$ terminal vertices, $\{F_1,F_2,...,F_n\}$, each corresponding to a single bit of output.  The number of vertices in the DAG is thus $h+n$.  We draw a directed edge $(T_i,T_j)$ iff for some input control $w_a$ of $TOF_j$, $LS(w_a)$ contains $Prod(\textsc{TOF}_i)$ with the non-zero coefficient $c_i$.  We draw a directed edge $(T_i,F_k)$ iff $LS(w_a)$, where $w_a$ is the piece of wire corresponding to the output $f_k$, contains $Prod(\textsc{TOF}_i)$. 

{\bf Edge labels.}  To each edge $(T_i,T_j)$ ending in the vertex $T_j$, we assign a label $EL_{T_i,T_j}$ with a numeric value from the set $\{1,2,3\}$.  The binary encoding of the label $EL_{T_i,T_j}$ tells which input-side controls of the gate $\textsc{TOF}_j$ contain term $Prod(\textsc{TOF}_i)$ in their $LS$ linear form representation.  Specifically, label $1=01_2$ says the second control of $\textsc{TOF}_j$ requires the knowledge of $Prod(\textsc{TOF}_i)$ to be computed ({\em i.e.}, $Prod(\textsc{TOF}_i)$ is included in the linear sum $LS$ of this piece of wire), label $2=10_2$ says the first control of $\textsc{TOF}_j$ contains $Prod(\textsc{TOF}_i)$ in its $LS$ form, and label $3=11_2$ says both controls of $\textsc{TOF}_j$ contain $Prod(\textsc{TOF}_i)$ in their $LS$ forms.  To each edge $(T_i,F_k)$ ending in the vertex $F_k$ we assign the numeric label of $1$.  The meaning of each such edge is the statement that the $Prod(\textsc{TOF}_i)$ is EXOR-ed with something else to obtain the output bit $f_k$, but it will become useful to think of the label as being equal to $1$, as opposed to any other number or no label, when counting the number of DAGs. 

{\bf Vertex labels.} Enumerate all $2^{n+1}$ linear functions of prime inputs $\{x_1,x_2,...,x_n\}$.  Each vertex $T_i$ is labelled by a set of two numbers, $VL_{T_i}:=(l_{w_{i,a}},l_{w_{i,b}})$, reporting the numeric orders of the linear functions of the primary inputs that are being EXORed to the $LS$ of the two input-side controls of the gate $TOF_i$, being the pieces of wire $w_{i,a}$ and $w_{i,b}$ directly feeding into the gate itself.  $l_{w_{i,a}} (l_{w_{i,b}})$ is obtained via removing all $Prod(\cdot)$ terms from $LS(w_{i,a}) (LS(w_{i,b}))$ and computing the numeric order of the respective linear function of the primary inputs.  Each vertex $F_k$ is labelled by the number $l_{F_k}$, corresponding to the numeric order of the linear function of primary inputs EXORed to the $k^{\text{th}}$ primary output.  It is obtained via removing all product terms from $LS(w_{F_k})$, where $w_{F_k}$ is the piece of wire carrying $k^{\text{th}}$ primary output.

Each such DAG uniquely defines a reversible function, and as such the number of different DAGs upper bounds the number of different reversible functions possible to obtain as a function of $h$---the number of Toffoli gates used.  Indeed, consider a specific instance of the above DAG, and obtain the reversible function it encodes.  We next construct logical functions computed in each vertex of the DAG, in the following order: $T_1,T_2,...,T_h,F_1,F_2,...,F_n$.  $T_1$ has no incoming edges and is labelled by $VL_{T_1}=(l_{w_{1,a}},l_{w_{1,b}})$.  The product computed by the $\textsc{TOF}_1$ gate is thus $Prod(T_1) = l_{w_{1,a}}(x) \& l_{w_{1,b}}(x)$; here, we use the numeric order of the linear function to call the function itself.  Incidentally, it is same as $LS(w_{1,a}) \& LS(w_{1,b})$ since this is the first Toffoli in the circuit.  $T_i$ has incoming edges that may be broken down into two sets $S_{T_i,1}$ and $S_{T_i,2}$, such that the edge label in each set has a non-zero digit $j=\{1,2\}$.  The product function $Prod(T_i)$ takes the value $(l_{w_{i,a}}(x)\oplus\bigoplus_{(T_k,T_i) \in S_{T_i,1}} Prod(T_k)) \& (l_{w_{i,b}}(x)\oplus\bigoplus_{(T_k,T_i) \in S_{T_i,2}} Prod(T_k))$.  The function assigned to $F_k$ is $Out(F_k):= l_{F_k}(x) \oplus \bigoplus_{(T_i,F_k) \in S_{F_k}} Prod(T_i)$, where $S_{F_k}$ is the set of edges coming into the vertex $F_k$.  The reversible function is given by the output vector $(Out(F_1),Out(F_2),...,Out(F_n))$.

We count the number of DAGs representing NCT circuits using the following product formula: the number of DAGs with vertex and edge labels equals to the product of the number of DAGs with edge labels and the number of ways to label vertices.  The second number is easy to obtain. Each vertex $T_i$ is labelled by a pair of linear Boolean functions of $n$ variables.  There are $2^{n+1}$ linear Boolean functions of $n$ variables, and as such the number of choices for the label of $T_i$ is $2^{2(n+1)}$.  The number of choices for the label of $F_k$ is $2^{n+1}$.  Given the total number of $T$ type vertices is $h$ and the number of $F$ type vertices is $n$, the overall number of vertex labels is $2^{2h(n+1)}\cdot 2^{n(n+1)}$. 

To count the number of DAGs with edge labels, describe those by a size $h \times (h+n)$ matrix $B=\{b_{i,j}\}_{i=1..h,j=1..h+n}, b_{i,j} \in \{0,1,2,3\}$, where $b_{i,j}=EL_{T_i,T_j}$, when $j \leq h$ and edge $(T_i,T_j)$ is in the DAG, $b_{i,j}=EL_{T_i,F_k}$, when $j>h$, $k=j-h$, and edge $(T_i,F_k)$ is in the DAG, and $b_{i,j}=0$ everywhere else.  Such matrices encode all DAGs, however, not every $h \times (h+n)$ matrix is being used by this encoding.  Specifically, $b_{i,h+k}$ ($k>0$) never take values above 1, all diagonal elements $b_{i,i}=0$, and, by construction of the DAG, its matrix $B$ has zeros below the diagonal.  Including those constraints gives an improved count compared to a simple count of all $h \times (h+n)$ matrices whose elements take one of $4$ values.  Specifically, the restricted set of matrices, subject to the above conditions, has $\frac{h(h-1)}{2}$ elements that may take any one of $4$ values and $nh$ elements that may take binary values.  The number of such restricted matrices is thus $4^{h(h-1)/2} \cdot 2^{nh}$. 

The number of DAGs, $4^{h(h-1)/2} \cdot 2^{nh} \cdot 2^{2h(n+1)}\cdot 2^{n(n+1)}$, should be at least as high as the number of reversible functions, $2^n!$.  Solving for $h$, we obtain:
\begin{eqnarray*}
2^{h(h-1)} \cdot 2^{nh} \cdot 2^{2h(n+1)}\cdot 2^{n(n+1)} \geq 2^n! \geq \sqrt{2\pi 2^n} \cdot 2^{n2^n} \cdot e^{-2^n} \\
h^2 + 3nh + h + n^2 + n \geq n2^n + C_12^n + C_2n \\
% h^2 \geq n2^n + o(n2^n) + o(h^2) \\
h \geq \sqrt{n}2^{n/2} + o(\sqrt{n}2^{n/2})
\end{eqnarray*}
Dropping lower degree additive terms, we obtain the desired inequality $\sqrt{n} 2^{n/2} \lesssim h = L_{0,0,1}(n)$. 
\end{proof}

\subsection{Upper bound}

Recall that every Boolean function can be written as a positive polarity Reed-Muller (PPRM) expansion, also known as Zhegalkin polynomial, 
\begin{equation}\label{eq:pprm}
f(x_1,x_2,...,x_n) = a_0 \oplus a_1x_1 \oplus a_2x_2 \oplus a_3x_1x_2 \oplus ... \oplus a_{2^n-1}x_1x_2...x_n,
\end{equation}
where $a_i|_{i=0..2^n-1}$ are Boolean numbers.  We will rely on the PPRM expansion in our construction.  In particular, we start by describing how one may obtain all product terms that the PPRM expansion relies on using optimal number of the Toffoli gates.

\vspace{2mm}\begin{lem}\label{lem:u1}
The set of all $2^n$ $n$-bit product terms $\{1, x_1, x_2, x_1x_2,..., x_1x_2...x_n\}$ may be generated by a reversible NCT circuit with the optimal number of $2^n-n-1$ Toffoli gates.
\end{lem} 
\begin{proof}
{\bf Lower bound.} The set $\{1, x_1, x_2, x_1x_2,..., x_1x_2...x_n\}$ contains $2^n-1$ linearly independent terms (all but first term are linearly independent).  The set of primary inputs, $\{x_1, x_2,...,x_n\}$, contains $n$ linearly independent terms.  For a set $S$ of Boolean functions the only way to obtain a new Boolean function that is linearly independent from all those in the set $S$ using NOT, CNOT, and Toffoli gates applied to those functions in the set, is to use a Toffoli gate.  As such, to generate $2^n-1$ linearly independent functions from the original set containing $n$ linearly independent functions, one must use at least $2^n-1-n$ Toffoli gates.  

{\bf Upper bound.} Denote $C(n)$ to be the number of Toffoli gates used to obtain the set of the product terms $\{1, x_1, x_2, x_1x_2,..., x_1x_2...x_n\}$. Once the set $\{1, x_1, x_2, x_1x_2,..., x_1x_2...x_n\}$ over $n$ variables is constructed, obtain the set $\{1, x_1, x_2, x_1x_2,...,x_{n+1},..., x_1x_2...x_{n+1}\}$ over $n+1$ variables as follows.  For each register $r$ in the existing set except first use the Toffoli gate with controls $r$ and $x_{n+1}$, and target residing in the value $0$ to compute the product $rx_{n+1}$ into the target bit.  This allows constructing the set of $2^n-1$ terms, $\{x_1x_{n+1}, x_2x_{n+1}, x_1x_2x_{n+1},...,$ $x_1x_2...x_nx_{n+1}\}$, at the cost of $2^n-1$ Toffoli gates.  Uniting these newly constructed terms with $\{1, x_1, x_2,$ $x_1x_2,..., x_1x_2...x_n\}$ that we already have and the input variable $x_{n+1}$ obtains the desired set $\{1, x_1, x_2,$ $x_1x_2,...,x_{n+1},..., x_1x_2...x_{n+1}\}$.  To summarise the above construction, we can write the following equality 
\begin{eqnarray*}
C(n+1) = C(n) + 2^n - 1.
\end{eqnarray*}
Observing that $C(1)=0$ allows solving this recurrence to obtain the desired $C(n)=2^n-n-1$.
\end{proof}

\vspace{2mm}\begin{thm}\label{thm:ub}
$L_{0,0,1}(n) \lesssim \frac{3}{\sqrt{2}}\sqrt{n} 2^{n/2}$.
\end{thm}
\begin{proof}
To obtain a circuit computing the reversible function $f(x_1,x_2,...,x_n)=(f_1(x_1,x_2,...,x_n),\linebreak f_2(x_1,x_2,...,x_n),...,f_n(x_1,x_2,...,x_n))$, we rely on the PPRM decomposition of the individual output components, followed by the grouping of variables into two non-overlapping sets $A$ and $B$, $A \sqcup B = \{x_1,x_2,...,x_n\}$, containing $a$ and $b$ variables each ($a+b=n$), as follows. Denote $P(\sigma_1,\sigma_2,...,\sigma_n) = x_{\sigma_{i_1}}x_{\sigma_{i_2}}...x_{\sigma_{i_k}}$, where $\sigma_i|_{i=1..n}$ are Boolean numbers and $\{\sigma_{i_1},\sigma_{i_2},...,\sigma_{i_k}\}$ is the set of all $\sigma_i=1$ within the set $\{\sigma_1,\sigma_2,...,\sigma_n\}$. Boolean $n$-tuples can furthermore be treated as natural numbers (via binary decomposition of integers).  
\begin{eqnarray*}
f_i(x_1,x_2,...,x_n) = \bigoplus_{j=0}^{2^n-1} P(j)f_i(j) \\ \nonumber
= \bigoplus_{j=0..2^a-1,f_i(j,\sigma)=1} P(j2^b) \& \Big( \bigoplus_{k=0..2^b-1} P(k)f_i(j,k) \Big) \\ \nonumber
= \bigoplus_{j=0..2^a-1,f_i(j,\sigma)=1} P(j2^b) \& \Big( \bigoplus_{k=0..2^b-1, f_i(j,k)=1} P(k) \Big)
\end{eqnarray*}
The circuit implementing $f(x_1,x_2,...,x_n)$ is obtained as follows:
\begin{enumerate}
\item Construct all positive polarity product terms over the set $B$ with $b$ variables. Per Lemma \ref{lem:u1}, this requires $2^b-b-1$ Toffoli gates. 
\item Construct all positive polarity product terms over the set $A$ with $a$ variables. Per Lemma \ref{lem:u1}, this requires $2^a-a-1$ Toffoli gates.
\item For each of the $2^a$ product terms in the set $A$, multiply this term by the proper function over the set $B$, obtained as a linear combination of product terms over the set $B$. Add the constructed terms together to obtain $i^{\text{th}}$ output bit.  This operation requires $2^a-1$ Toffoli gates per each of the $n$ bits of the target function $f(x_1,x_2,...,x_n)$; this is because we use a CNOT instead of the Toffoli gate to multiply by term $1$.  The total Toffoli gate count of this part is thus $n(2^a-1)$.
% FIXME: This operation requires $2^a-1$ Toffolis (Use CNOT instead of Toffoli to multiply by term 1, if needed) per each of the $n$ bits
\end{enumerate}
The total Toffoli gate count in the above construction is $2^a+2^b-a-b-2+n(2^a-1)$.  Assigning $a:=\frac{n-\log_2n}{2}$ and $b:=\frac{n+\log_2n}{2}$, we furthermore obtain:
\begin{eqnarray*}
2^\frac{n-\log_2n}{2} + 2^\frac{n+\log_2n}{2} - n - 2 + n(2^\frac{n-\log_2n}{2}-1) = \frac{2^{n/2}}{\sqrt{n}} + \sqrt{n}2^{n/2} - 2n - 2 + \sqrt{n}2^{n/2} \\ \nonumber
= 2\sqrt{n}2^{n/2} + o(\sqrt{n}2^{n/2}). \nonumber
\end{eqnarray*}
The above calculation relies on the real-valued $a$ and $b$, whereas in our construction numbers $a$ and $b$ must take integer values.  This limitation imposes the requirement to correct the leading coefficient by the ratio $\frac{\min\{f(0),f(1)\}}{\min_{x \in [0,1]}f(x)}$, where $f(x)=2^x+2^{1-x}$, and $x$ plays the role of the fractional part of either $a$ or $b$.  This ratio equals to $\frac{3}{2\sqrt{2}}$, resulting in the overall upper bound of
\begin{eqnarray*}
L_{0,0,1}(n) \lesssim \frac{3}{2\sqrt{2}}2\sqrt{n} 2^{n/2} = \frac{3}{\sqrt{2}}\sqrt{n} 2^{n/2}. \nonumber
\end{eqnarray*}
\end{proof}

\begin{cnj}
$L_{0,0,1}(n) \lesssim \sqrt{n} 2^{n/2}$.
\end{cnj}

\subsection{Corollaries and discussion}

Define the non-Clifford cost of a quantum circuit to be the number of operations outside the Clifford group that it contains.  $T$-count, a metric of this kind, is popular in applications, owing to the dominating cost of the $T/T^\dagger$ gates over the cost of other gates.

\vspace{2mm}\begin{cor}\label{cor:1}
The $T$-count of quantum circuits implementing a reversible function $f(x)$ of $n$ primary inputs in the form of the mapping $(x,0) \mapsto (x,f(x))$ can be upper bounded by the expression $21\sqrt{n}2^{n/2} + o(\sqrt{n}2^{n/2})$. 
\end{cor}

\begin{proof}
The desired construction relies on the Bennett's trick \cite{ar:b}.  The Bennett's trick ensures that all auxiliary bits are properly cleaned and no residual entanglement remains that may prohibit from using the desired reversible circuit within quantum algorithms.  In particular, apply the result of Theorem \ref{thm:ub} to obtain a reversible NCT circuit with $n$ Boolean outputs $(f_1(x),f_2(x),...,f_n(x)) = f(x)$.  This circuit, $C$, relies on $2^a+2^b-n-2+n(2^a-1)$ Toffoli gates and computes functions $f_1,f_2,...,f_n$, product terms over the set $A$, and product terms over the set $B$.  To obtain the desired mapping $(x,0) \mapsto (x,f(x))$, we only need to uncompute product terms over the sets $A$ and $B$ using the inversion of the circuit that was used to compute them.  Such a circuit uses $2^a+2^b-n-2$ Toffoli gates.  The overall gate count is thus $2(2^a+2^b-n-2)+n(2^a-1)$.  Equating $2\cdot 2^b = 2 \cdot 2^{n-a}$ to $n2^a$ allows to obtain favourable asymptotic. This requires the parameter $a$ to take the value $\frac{n+1-\log_2n}{2}$.  The overall Toffoli gate count corrected for the notion that $a$ must be integer is thus $\frac{3}{2\sqrt{2}}\cdot 2\sqrt{2}\sqrt{n} 2^{n/2} + o(\sqrt{n}2^{n/2}) = 3\sqrt{n}2^{n/2} + o(\sqrt{n}2^{n/2})$. Since each Toffoli gate relies on seven $T/T^\dagger$ gates, the overall $T$-count is $21\sqrt{n}2^{n/2} + o(\sqrt{n}2^{n/2})$. 
\end{proof}

\begin{table}[ht]
\centerline{
\begin{tabular}{ll|ll} \hline \hline
n 	& $T$-count & n 	& $T$-count	 	\\ \hline
3   & 36 		& 12 & 4,648  \\
4   & 84		& 13 & 6,643  \\
5   & 175 		& 14 & 10,430 \\
6   & 294 		& 15 & 14,441 \\
7   & 525 		& 16 & 22,036 \\
8   & 812 		& 17 & 30,975 \\
9   & 1,295 	& 18 & 46,186 \\
10  & 2,002 	& 19 & 65,877 \\
11  & 2,989 	& 20 & 96,320\\ \hline
\end{tabular} 
}
\caption{Upper bounds on the $T$-count of an arbitrary reversible $f(x)$ over $n$ variables, implemented as the mapping $(x,0)\mapsto(x,f(x))$ with the use of arbitrary ancillae.} \label{tab:main}
\end{table}

Table \ref{tab:main} reports upper bounds on the number of $T/T^\dagger$ gates in the NCT circuit realizations of reversible $n$-bit functions for small values $n$.  

In our constructions of the upper bounds we relied on the seven $T/T^\dagger$ gate implementation of the Toffoli gate.  However, in the presence of measurements and the ability for classical feedback, the Toffoli gate can be implemented via a circuit with only four $T/T^\dagger$ gates \cite{ar:j}.  This means that the upper bound in Corollary \ref{cor:1} drops down to $12\sqrt{n}2^{n/2} + o(\sqrt{n}2^{n/2})$, and all $T$-counts in Table \ref{tab:main} can be reduced by the factor of $7/4$ ({\em e.g.}, a $15$-bit reversible function would require only at most $8252$ $T/T^\dagger$ gates). 

In the above, we upper bounded the $T$-count cost of the implementations of reversible functions by an expression of the form $O(\sqrt{n}2^{n/2})$, as well as reported a table showing  the $T$-count for small numbers of inputs $n$ (Table \ref{tab:main}).  We will next consider a more realistic circuit cost metric and establish a lower bound on its value to show that the use of the $T$-count may significantly downplay the real cost of circuit implementations.  The lesson here is the $T$-count metric must be used with extra care, or better yet replaced with a metric that does not lead to an abuse of a resource deemed less costly and thereby ignored, such as the $T$-count does with the CNOT gates. 

Per previous constructions, the number of Toffoli gates that suffice to implement an arbitrary $n$-bit reversible function is upper bounded by the expression $\min_{\{a>0,b>0,a+b=n\}}{(2(2^a+2^b-n-2)+n(2^a-1))}$, and thereby the $T$-count is no more than 
$$X:=7\cdot \min_{\{a>0,b>0,a+b=n\}}{(2(2^a+2^b-n-2)+n(2^a-1))},$$
 in the scenario when we are concerned with the potential unwanted entanglement.  Recall that when applying the Bennett's trick to this circuit, the sub-functions $\bigoplus_{k=0..2^b-1, f_i(j,k)=1} P(k)$ need to be computed and multiplied by a proper product the over variable set $A$ only in the first part of the circuit, but are unnecessary in the second part, as they are uncomputed after each use.  Such circuits implementing the reversible functions use no more than a total of $S=2^a+2^b+n+1$ bits: $2^a$ bits contain products over the set $A$ (including primary inputs), plus $2^b$ bits containing products over the set $B$ (including primary inputs), plus $n$ bits where the output values $f_1,f_2,...,f_n$ are constructed, plus $1$ bit to compute/uncompute different sub-functions $\bigoplus_{k=0..2^b-1, f_i(j,k)=1} P(k)$.  Next, establish a lower bound on the quantity $L_{0,1,1}(n,S-n)$, counting the number of CNOT and Toffoli gates in the reversible circuits over $S$ bits and implementing reversible $n$-bit functions.  Considering the quantity $L_{0,1,1}(n,S-n)$ ensures we use same number of ancillary qubits as that used to obtain the number $X$. Applying \cite[Lemma 8]{ar:spmh}, a lower bound is given by the expression 
 $$Y:=\frac{\log_2{G}}{\log_2{b}} = \frac{\log_2(2^n!/2^n)}{\log_2(4S(S-1)(S-2)+4S(S-1))} (= \frac{2^{n+1}}{3} + o(2^n)).$$  

To be able to compare the numeric values of $X$, being the upper bound expressing the $T$-count to $Y$, being the lower bound for the number of CNOT/Toffoli gates, divide $Y$ by the cost of the $T$ gate expressed in terms of the cost of the CNOT, being the cheaper one between the CNOT and the Toffoli.  We have previously established that this number may carry a value of about $50$.  Comparing $X$ to $Y/50$ reveals that for $n=27, a=12$ the latter is already greater than the former.  This means that the $T$-count cost metric may already undervalue the real cost of the circuits when $n$ is as small as $27$.  By the time $n=50$ ($a=23$), the difference between the two grows to a factor of $2662$, meaning that for the numbers this high the $T$-count cost metric can be rather misleading.  The discrepancy furthermore grows very rapidly with $n$---specifically, with the speed $C\frac{2^{n/2}}{\sqrt{n}}$, for some constant $C$.

While the above numbers clearly discourage from the use of the $T$-count circuit metric in scalable designs, we suspect that the real scope of the potential misinformation carried by using the $T$-count may be much larger.  This is because in our calculations we did not account for such resources as the cost of ancilla, or the cost of the long-range CNOT gates, that are downplayed (in fact, ignored) by the $T$-count.  On the other hand, we proved that the discrepancy exists in general, whereas practical quantum computations rely on very specific and well-structured reversible transformations (such as arithmetic circuits, including exponentiation part of Shor's algorithm).  The extent to which the discrepancy can and does manifest itself in practice and over such structured circuits needs to be studied separately.

\section{Summary of the results}
Our study details reversible NCT circuit complexity figures by the gate types, leading to the following list of refined optimal and asymptotically optimal values for the respective counts. 
\begin{itemize}
\item[000.] $\forall g$ $L_{0,0,0}(n,g) = 0$;
\item[001.] $\sqrt{n}2^{n/2} \lesssim L_{0,0,1}(n) \lesssim \frac{3}{\sqrt{2}}\sqrt{n}2^{n/2}$;
\item[010.] $\forall g$ $L_{0,1,0}(n,g) = 0$;
\item[011.] $\frac{n2^n}{3 \log_2{n}} \lesssim L_{0,1,1}(n,1) \lesssim \frac{40n2^n}{\log_2{n}}$;
\item[100.] $\forall g$ $L_{1,0,0}(n,g)=0$;
\item[101.] $\sqrt{n}2^{n/2} \lesssim L_{1,0,1}(n) \lesssim \frac{3}{\sqrt{2}}\sqrt{n}2^{n/2}$;
\item[110.] $L_{1,1,0}(n,1)=1,$ $\forall g>1$ $L_{1,1,0}(n,g)=0$;
\item[111.] $\frac{n2^n}{3 \log_2{n}} \lesssim L_{1,1,1}(n,1) \lesssim \frac{48n2^n}{\log_2{n}}$; 
\end{itemize} 

\section{Conclusion}
In this paper, we studied the complexity function $L_{a,b,c}(n,g)$, detailing reversible NCT circuit costs by the gate types used.  We established asymptotically optimal or optimal counts in every possible scenario.  Of these, some bounds were known from the previous literature.  We upper and lower bounded the multiplicative complexity of reversible circuits, leading to their asymptotic optimality.  We formulated a conjecture stating that $L_{0,0,1}(n) \lesssim \sqrt{n}2^{n/2}$.  Proving this conjecture would establish that the multiplicative complexity of reversible functions is equal to $\sqrt{n}2^{n/2}$ up to lower order additive terms.  We furthermore applied our study to show the limitations on the use of the $T$-count, multiplicative complexity, and Toffoli count metrics in practical designs.  The discrepancy between a real cost and the one provided by the $T$-count/multiplicative complexity/Toffoli count may be as high as $C\frac{2^{n/2}}{\sqrt{n}}$, where $C$ is a constant.  Taking some realistic parameters we estimated that for $n=50$ the $T$-count may misrepresent a real cost of the circuit it is applied to evaluate by a factor of as much as $2662$. 

\section*{Acknowledgements}

I wish to thank Prof. Sergey B. Gashkov from Lomonosov Moscow State University and Dr. Martin R\"{o}tteler from Microsoft Research for their helpful discussions. 

This material was based on work supported by the National Science Foundation, while working at the Foundation. Any opinion, finding, and conclusions or recommendations expressed in this material are those of the author and do not necessarily reflect the views of the National Science Foundation.

\end{document}